\documentclass{amsart}
\usepackage{longtable}

\usepackage{pifont}
\usepackage[round]{natbib}
\usepackage{amsfonts, amssymb, mathrsfs}
\usepackage{amsmath}
\usepackage{amsthm}
\usepackage{graphicx}
\usepackage{color}
\usepackage{subfig}
\usepackage{multirow}
\usepackage[all]{xy}
\usepackage{ifpdf}

\ifpdf
\else
\fi
\usepackage{url}

\newtheorem{theorem}{Theorem}[section]
 \newtheorem{lemma}[theorem]{Lemma}
\newtheorem{definition}[theorem]{Definition}
\newtheorem{proposition}[theorem]{Proposition}
\newtheorem{remark}[theorem]{Remark}
\newtheorem{problem}[theorem]{Problem}
\newtheorem*{question}{Question}%[section]
\newtheorem{notation}[theorem]{Notation}
\newtheorem{exampleth}[theorem]{Example}
\newenvironment{example}{\begin{exampleth}}{\hfill $\diamond$\\ \end{exampleth}}
\newcommand \ZZ {\mathbb{Z}}

\newcommand \pp {\mathcal{P}}

\newcommand \RR {\mathbb{R}}

\newcommand \codim{\text{codim}}

\newcommand \bb {\mathcal{B}}

\newcommand \Sn {\mathbb{S}}
\newcommand \cT {\mathscr{T}}
\newcommand \cD {\mathscr{D}}

\newcommand \gfan {\texttt{{G}fan}}
\newcommand \polymake {\texttt{{P}olymake}}
\newcommand \PP {\mathcal{P}}

\newcommand \cc {\mathscr{C}}
\newcommand \fastme {\texttt{{F}astme}}

\newcommand \ww {\omega}
\newcommand \vvv {\nu}
\newcommand \cB {\mathcal{B}}

\newcommand{\restr}[1]{\left|_{#1}\right.}

\newcommand \tofrom {\leftrightarrow}
\newcommand{\eat}[1]{}
\newcommand{\RF}{\ensuremath{\delta_{RF}}}
\newcommand{\splits}{\Sigma}
\newcommand{\shift}{h}
\newcommand{\ssa}{\mathsf{a}}
\newcommand{\ssb}{\mathsf{b}}
\newcommand{\ssc}{\mathsf{c}}
\newcommand{\ssd}{\mathsf{d}}

\makeatletter
\@namedef{subjclassname@2010}{ \textup{2010} Mathematics Subject Classification} \makeatother

\begin{document}

\begin{abstract}
It is well known among phylogeneticists that adding an extra taxon (e.g. species) to a data set can alter the structure of the optimal phylogenetic tree in surprising ways.
However, little is known about this ``rogue taxon'' effect.
In this paper we characterize the behavior of balanced minimum evolution (BME) phylogenetics on data sets of this type using tools from polyhedral geometry.
First we show that for any distance matrix there exist distances to a ``rogue taxon'' such that the BME-optimal tree for the data set with the new taxon does not contain any nontrivial splits (bipartitions) of the optimal tree for the original data.
Second, we prove a theorem which restricts the topology of BME-optimal trees for data sets of this type, thus showing that a rogue taxon cannot have an arbitrary effect on the optimal tree.
Third, we construct polyhedral cones computationally which give complete answers for BME rogue taxon behavior when our original data fits a tree on four, five, and six taxa.
We use these cones to derive sufficient conditions for rogue taxon behavior for four taxa, and to understand the frequency of the rogue taxon effect via simulation.
\end{abstract}

\title[Polyhedral Geometry of Phylogenetic Rogue Taxa]{Polyhedral Geometry of\\ Phylogenetic Rogue Taxa}
\author{Mar\'ia Ang\'elica Cueto}
\email{macueto@math.berkeley.edu}
\address{Department of Mathematics,
  University of California, Berkeley, CA 94720, USA.}

\author{Frederick A. Matsen}
\email{matsen@fhcrc.org} 
\address{
Program in Computational Biology,
Fred Hutchinson Cancer Research Center,
1100 Fairview Ave. N. M1-B514,
P.O. Box 19024,
Seattle, WA 98109-1024
 } 
\thanks{The first author was
  supported by a UC Berkeley Chancellor's Fellowship. The second
  author was supported by the Miller Institute for Basic Research at UC Berkeley.}


%\date{\today}
\maketitle

\section{Introduction}
\label{sec:introduction}

Ideally, phylogenetic data sets would have the property that the optimal tree for a subset $X$ of taxa $Y$ would be the same as the tree obtained by restricting the optimal tree on $Y$ to the set $X$.
However, practicing phylogeneticists are well aware that this is not the case; the extensive literature on ``taxon sampling'' reviewed below is evidence to the contrary.
One can also find references to ``rogue taxa'' which, although not clearly defined or rigorously investigated, are taxa who do not fit into a tree and whose inclusion may disrupt the inference of evolutionary relationships of the other taxa.
For example, \citet{sullivanSwoffordGuineaPigsRodents97} state ``\dots the hedgehog therefore appears to represent a `rogue' taxon that cannot be placed reliably with these data and that possibly confounds attempts to estimate the relationships among the remaining taxa.''
The ``rogue'' descriptor is also used by \citet{baurainEaLackOfResolution07} to describe taxa with a ``strong nonphylogenetic signal''; these authors describe the importance of finding and eliminating these taxa from phylogenetic studies. 

Surprisingly, we were unable to find any mathematical or simulation-based analysis of the action of rogue taxa in phylogenetic trees. 
The closest studied subject is ``taxon sampling.'' 
This area of research is focused on the following question: if we are interested in the phylogenetic tree on a set of taxa $Y$, do we do better or worse by adding more taxa into the tree?
If better, is the improvement more significant than would be gained by increasing the length of the sequences (by redirecting resources)? 

The origins of the taxon sampling debate can be traced to the pioneering paper of \citet{felsenstein78} that demonstrated mathematically the existence of ``long branch attraction,'' where two pendant branches are artifactually placed close together by parsimony algorithms.
This led to the question of if parsimony long branch problems could be dispensed with by adding new taxa to the dataset to break up the long branches; \citet{hendyPennyFramework89} have answered affirmatively under certain conditions.
The investigation was continued by \citet{kimInconsistencyParsimony96}, who showed that the situation is subtle and that the new taxa must appear in specific regions of the tree in order to counter the long branch attraction problem.

These mathematical investigations of parsimony were followed by a
flood of simulation-based papers investigating maximum likelihood,
parsimony, as well as distance methods for phylogenetics.
\citet{hillisInferringComplex96}, \citet{graybealBetter98}, and
\citet{poeSensitivity98} indicated that a larger number of taxa
improved estimation, whereas the high-profile publication of
\citet{rosenbergKumarIncompleteTaxSampOK01} claimed the opposite.
  The
Hillis group responded 
\citep{zwicklHillisIncreasedTaxSampGood02,pollockEaTaxSampGood02,hillisEaIsSparseTaxSampProblem03}
which led to \citet{rosenbergKumarGenomics03}  somewhat moderating their
position.
  The debate on taxon sampling has continued to the
present day, with additional simulations
\citep{poeLongBranchSubdivision03, debrySamplingParsimony05,
  hedtkeEaResolutionByTaxSamp06}, review articles
\citep{heathEaTaxSampAccuracy08}, and studies to understand the impact
of taxon sampling on the inference of macroevolutionary processes \citep{heathEaTaxSampMacroevol08}.
The simulation literature in this area is considered important enough
to even have a paper \citep{rannalaEaTaxSampAccuracy98} about
methodology for taxon-sampling simulations.

There are two inherent difficulties with simulations of this type.
First, the collection of possible parameter values for simulation is vast, and any simulation study must make choices about which parameters to use.
This first problem alone may be the source of the disagreement found in the taxon selection literature.
Second, the simulations are done by simulating data with a single model on a tree, then reconstructing.
This does not address the problem of what happens when considering unusual data sets, such as those obtained by major model misspecifications.

A mathematical approach can address these difficulties, although with
certain caveats.
Theorems can indicate that a phenomenon will always happen given certain criteria, and the construction of the complete spaces of examples or counter-examples gives very precise information about these questions.
  By exploring the complete space of
data sets of a certain type, one is not limited to data sets which are
within a certain class of models.
  The trade-off for the strength of
these conclusions is that often the setting must be simplified to make the
problem mathematically tractable.

In order to address taxon selection and the rogue taxon effect problem mathematically, we have chosen to use distance-based phylogenetics, specifically the Balanced Minimum Evolution (BME, described below) criterion.
Because the optimality criterion is expressed in terms of the minimization of an inner product, we are able to harness the power of polyhedral geometry to answer the questions of interest with a high degree of precision. 
Although BME-based algorithms are not among the most popular in
phylogenetics, implementations do exist which show good performance under simulation \citep{desperGascuelBMEFastAccurate02}. 
The BME criterion is consistent \citep{desperGascuelBME04}, as is FastBME which minimizes BME through tree rearrangements \citep{bordewichEaFastBMEConsistent09}. 
Another motivation for studying BME is the close relationship between BME and the very popular Neighbor-Joining (NJ) algorithm \citep{saitouNeiNJ87}.
Specifically, NJ has been shown to be a heuristic BME minimizer
\citep{desperGascuelME05}; the relationship between the two algorithms
has been investigated by \citet{BMEPolytope}. 

After describing a bit of terminology, we will discuss the main results of the paper.
Note that by \emph{dissimilarity map} we simply mean a mapping $D$ from unordered pairs of taxa to non-negative numbers such that $D(x,x) = 0$ for all $x$.
These are sometimes called ``distance matrices'' in the phylogenetics literature but we use dissimilarity map to emphasize that they need not satisfy the triangle inequality.

\begin{definition}
  Let $t$ be a phylogenetic tree equipped with branch lengths
  $\mathbf{b}$.  The tree metric associated with $t$ and $\mathbf{b}$ is the
  dissimilarity map obtained as follows: the distance between taxa $i$
  and $j$ of $t$ is given by the total length (i.e.\ sum of branch lengths)
  of the path from $i$ to $j$ in $t$ with respect to $\mathbf{b}$.
\end{definition}

Next we define some core objects of study for this paper.

\begin{definition}
  \label{def:DandLifting} 
  Let $D$ be a dissimilarity map on $n$ taxa. 
  A ``lifting'' $\tilde{D}$ of $D$ is a dissimilarity map on $n+1$ taxa obtained from $D$ by adding distances from the first $n$ taxa to an $(n+1)$st taxon.
\end{definition}

\begin{definition}
  Let $D$ be a dissimilarity map on $n$ taxa, and let $\tilde{D}$ be a lifting of $D$. 
  The BME tree for $D$ will be called the ``lower tree,'' while the BME tree for $\tilde{D}$ will be called the ``upper tree.''
  The ``restricted upper tree'' will be the tree induced on the original $n$ taxa by restricting the upper tree to this set.
\end{definition}

Our primary goal is to understand topological differences between the upper and lower trees for various original dissimilarity maps $D$ and various liftings $\tilde{D}$.

\subsection{Overview of the paper}
\label{sec:overview-paper}
The first section describes the effect of adding a new taxon when the original dissimilarity map $D$ is arbitrary.
Theorem~\ref{thm:funny} shows that for any $D$ there exists a lifting such that the intersection of the split sets for the restricted upper tree and the lower tree consists of the trivial pendant splits.
  In other words, we show that the restricted
upper tree and the lower tree can be maximally distant in terms of the
Robinson-Foulds metric \citep{robinsonFouldsComparison81}.
  However,
the upper tree cannot deviate from the lower in an arbitrary way:
Theorem~\ref{thm:NeverSx} shows that certain combinations of lower and
upper trees are not possible.
We also note that the trees of Theorem~\ref{thm:funny} need not be maximally distant in terms of the quartet distance (Remark~\ref{rmk:quartet}).

The second section addresses the case when the original dissimilarity map $D$ is a tree metric for
some tree $t$; in this setting there is no question of what the
optimal tree for the lower taxa ``should'' be.
  That is, if the upper
tree does not contain the lower tree, the additional taxon is
definitely a disrupting ``rogue'' taxon.
  When $D$ is a tree metric,
there exists a simplified formulation of the BME computations.
  This
``reduced'' formulation has a linear rather than a quadratic number of
variables, and allows polyhedral computation directly over the
parameters of interest.
  We study the associated ``reduced polytope''
and several of its combinatorial and geometric properties, including
its dimension.
  Using this ``reduced'' formulation we are able to give
sufficient conditions (Propositions~\ref{prop:rogueFunnyBL}
and \ref{prop:rogueThreeTimes}) for the rogue taxon effect when the
lower tree has four taxa, as well as a perspective on the frequency of
the rogue effect through simulations for up to six lower taxa.

The computations in this paper were done with a combination of
\gfan~\citep{gfan}, \polymake~\citep{polymake}, and custom
\texttt{ocaml} code using GSL, the GNU scientific library\nocite{ocaml}\nocite{gsl}. 
For the interested reader, source code is
available at 
\begin{center}
  \texttt{\url{http://github.com/matsen/roguebme}}.
\end{center}

%%%%%%%%%%%%%%%%%%%%%%%%%%%%%%%%%%%%%%%%%%%%%%%%%%%%%%%%%%%%%%%%%%%%%%%%%%%%%
%%%%%%%%%%%%%%%%%%%%%%%%%%%%%%%%%%%%%%%%%%%%%%%%%%%%%%%%%%%%%%%%%%%%%%%%%%%%%

\section{Polyhedral geometry and BME phylogenetics}
\label{sec:polyhedral-geometry}

In this section, we introduce the mathematical problem we wish to investigate and walk through the necessary background in polyhedral geometry.
We start by defining the Balanced Minimum Evolution (BME) criterion for phylogenetic inference. 

For the purposes of this paper, all trees will be unrooted phylogenetic trees.
We will use parenthetical ``Newick format'' to describe trees, such that $((a,b),(c,d),e)$ indicates a five taxon tree with the pairs $a,b$ and $c,d$ being sister taxa \citep{felsensteinBook}.
Sometimes we will write these unrooted trees in a rooted manner, as we feel that $((a,b),(c,d))$ is clearer than $(a,b,(c,d))$.
The degree-two vertex of the rooted representation should be suppressed.
\emph{Trivalent} trees are trees such that all internal nodes have degree three.
\begin{definition}
  Given a dissimilarity map $D$ in $\RR^{\binom{n}{2}}$, the
  ``{Balanced Minimum Evolution}'' (BME) length of a phylogenetic tree
  $t$ with respect to a dissimilarity map $D$ is the quantity
  \begin{equation}
    \lambda(t,D) := \sum_{1\leq i<j\leq n} \ww^t_{ij} D_{ij},\label{eq:BMECriterion}
  \end{equation}
  where $\ww_{ij}^t=\prod_{v \in p_{ij}^t} (\deg(v)-1)^{-1}$, and $p_{ij}^t$ denotes the internal vertices in $t$ on the path between leaves $i$ and $j$.
\end{definition}
\begin{remark}
  In the case of a trivalent tree $t$, the weight $\ww^t_{ij}$
  equals $2^{-|p_{ij}^t|}$.
\end{remark}

A \emph{BME tree} for an $n \times n$ non-negative matrix $D$ will be a tree
$t$ minimizing $\lambda(t,D)$ over all $n$-taxon trees.
  The BME algorithm is \emph{consistent} on trivalent trees: if $D$ is tree
metric with trivalent tree topology $t$, then the BME tree of $D$ is
$t$ \citep{desperGascuelBME04}.

%Global BME minimization is known to be hard \citep{bmeHard}.
Note that there is a volume-zero set of dissimilarity maps with
multiple optimal BME trees, and therefore it is not quite right to
speak of ``the'' BME tree.
  All of our statements are true by
replacing ``the BME tree'' with ``a BME tree'', however, we prefer
stating the former.
 More precisely, given a dissimilarity map, we have
two cases: either the set of a possible BME trees of $D$ consist of a
single (trivalent) tree, or the set has size at least two and it is
closed under degenerations.
 That is, if a trivalent tree $t$ contracts to a BME tree for $D$, then $t$ is also a BME tree for $D$; this claim will be clear from the polyhedral perspective described below.

There are several equivalent formulations of the BME length
\citep{BMEPolytope}, although we prefer \eqref{eq:BMECriterion} because of its polyhedral interpretation.

%  Note that as the BME criterion is an inner product, we can multiply all distances by a constant factor and not change the order of BME lengths for a collection of phylogenetic trees.
%  We will sometimes work with larger distances than one would typically encounter from data.
%  If this seems bothersome for the reader, then simply scale the distances down.

Global BME minimization is known to be hard \citep{bmeHard}.
 The widely used Neighbor-Joining algorithm approaches the BME problem from a greedy perspective \citep{NJGreedyBME}.
The \fastme\ algorithm starts with a heuristically obtained tree and then refines it using Nearest Neighbor Interchange (NNI) to attempt to find the BME minimal tree \citep{desperGascuelBME04}.  
A better understanding of the BME polytope (defined below) could lead to better such algorithms \citep{FastMEP}, analogous to how understanding the traveling salesman polytope provides insight into the traveling salesman problem \citep{MR811476}.

We now introduce the BME polytope, first investigated by \citet{BMEPolytope}.
A \emph{polytope} in $\RR^m$ is the convex hull of a finite number of points in $\RR^m$.
Fix a positive integer $n$. 
The \emph{BME polytope} in $\RR^{\binom{n}{2}}$ is the convex hull of the points $(\ww^t_{ij})_{i,j}$, where $t$ varies among all possible tree topologies on $n$ taxa.
Using this polyhedral interpretation, the problem of finding the
BME-optimal tree $t$ on $n$ taxa corresponds to picking a vertex $\ww^t$ of the BME polytope minimizing the Euclidean dot product of the vertex with a given dissimilarity map (considered as a vector in $\RR^{n \choose 2}$).
The BME tree is the tree associated to this vertex.

We can characterize this optimization process by constructing the corresponding \emph{inner normal fan}.
The inner normal fan of a polytope $\pp\subset \RR^N$ is given as a
finite collection of cones (i.e.\ a set closed under multiplication by positive scalars) as follows.
Each cone in the inner normal fan of $\pp$ corresponds to a face $\mathcal{F}$ of the polytope $\pp$ and is defined as
\begin{equation}
\cc_{\mathcal{F}}:=\{w\in \RR^N\colon
\langle w ,v\rangle =\min\{\langle w , u \rangle : u\in \pp\},\;  \forall \, v\in F\},
\end{equation}
i.e.\ those vectors such that the minimum inner product is achieved at all points of the face $\mathcal{F}$.

By construction, each cone is polyhedral: it is the solution set of a
system of linear inequalities.  As such, it can be expressed as the
positive span (i.e.\ using nonnegative scalars) of finitely many vectors,
which we call \emph{extremal rays}.  In addition, the inner normal fan
of $\pp$ is a \emph{polyhedral fan} because the family
$\{\cc_{\mathcal{F}}\colon \mathcal{F}\subset \pp \text{ face}\}$ is
closed under intersections.  Moreover, this fan is \emph{complete}
(i.e.\ the union of all cones equals the ambient space $\RR^N$) and
each cone $\cc_{\mathcal{F}}$ has dimension equal to $\codim\,
\mathcal{F}=N-\dim \mathcal{F}$, where $\dim \mathcal{F}$ denotes the
dimension of the affine span of face $\mathcal{F}$.  In particular, if $\mathcal{F}$ is a
vertex, then $\cc_{\mathcal{F}}$ is full dimensional.  We call these
full-dimensional cones \emph{chambers}.  The inner normal fan of the
BME polytope will be referred to as the \emph{BME fan}.  We refer the
reader to \citep[][Chapter 1]{Ewald} for a complete exposition of
normal fans.
\begin{remark}
  From the previous discussion we see that the BME criterion is
  equivalent to the membership of a dissimilarity map $D$ to a
  \emph{chamber} in the \emph{BME fan}.
 Thus $D$ belongs to the interior of a chamber in the BME fan \emph{if and only if} the BME tree of $D$ is unique.
 The boundary of these chambers is the volume zero set having multiple BME trees (discussed earlier in this section).
\end{remark}

Since the BME polytope encodes the problem of finding the BME tree of
a dissimilarity map, it is worth understanding its structure. 
Some of its combinatorial properties have been studied for small number of taxa, although several questions remain open for $n\geq 6$.
  We investigate some of its features below, as described by \citep{BMEPolytope}.

The vertices of the BME polytope correspond to the points $(\ww^t_{ij})_{i,j}$ where $t$ is a \emph{trivalent} tree, for a total of $(2n-5)!!$ vertices \citep[Lemma 2.33]{ASCB}.
Here, $(2n-5)!!=(2n-5)\cdot (2n-3)\cdots 3\cdot 1$.
In addition, the vector $\ww^s_{ij}$ associated to the star tree $s$ (the tree with a single internal node) lies in the interior of the polytope, whereas all other
points $\ww^t$ lie on its boundary \citep[Lemma 2.1]{BMEPolytope}.

The \emph{dimension} of the BME polytope (i.e.\ the dimension of the affine space spanned by this polytope) is $\binom{n}{2}-n$.
  The polytope is  not full-dimensional because, after translation to  the origin, the orthogonal complement of its affine span is spanned by the $n$ \emph{shift vectors} $\{\shift_{\mathbf{a}}\colon \mathbf{a}\in \{1, \ldots, n\}\}$.
Here, the shift vector $\shift_{\mathbf{a}}$ refers to a dissimilarity map in which leaf $\mathbf{a}$ is at distance 1 from all other leaves,
while all other pairwise distances are $0$.

The \emph{$f$-vector} $\mathbf{f}(\PP)\subset \RR^N$ of an
$N$-dimensional polytope $\PP$ gives the number of faces of each
dimension of $\PP$.  That is, $\mathbf{f}(\PP)_i = \#\{$faces of
dimension $i-1$ of $\PP\}$.  The $f$-vectors of BME polytopes have
been studied for up to seven taxa.  In particular, for four and five
taxa, these vectors have been completely described in \cite[Table
1]{BMEPolytope}, whereas for six and seven taxa some of the entries of
the $f$-vector have remained unknown up to now.  We were able to
compute the complete $f$-vector for six taxa by methods of tropical
geometry, using \gfan.
% The lineality space is 6 dimensional, and maximal faces are 13 dimensional.
% After moding out by the lineality space they are 7 dimensional.
The resulting $f$-vector is:
\[
%(90262, 640140, 1742445, 2373345, 1715455, 635265, 105945, 5460).
(105, 5460, 105945, 635265,  1715455, 2373345, 1742445, 640140, 90262).
\]
In particular, we see that the polytope has $90262$ facets.
It also has $105$ vertices, labeled by all trivalent trees on six taxa.

As a corollary of these computations, it follows that the edge graph
of the BME polytope for six taxa is the complete graph $K_{105}$ \citep{BMEPolytope}.
This says that any two vertices of the BME polytope can be connected
by an edge.
 Similar behavior occurs for four and five taxa, but this
is no longer true for seven or more taxa \citep{BMEPolytope}.

By construction, the BME polytope comes equipped with a natural
symmetry given by the symmetric group $\Sn_n$ on $n$ elements.
Namely, relabeling the leaves of a trivalent tree $t$ by a permutation
$\sigma\in \Sn_n$ sends $t$ to the relabeled trivalent tree $\sigma
t$, and hence the vertex $\ww^t$ to $\ww^{\sigma{t}}$.  In a similar
way, higher dimensional faces of the BME polytope will have this
symmetry.  Therefore, we can encode these symmetries in the
$f$-vector, and record the number of faces of each dimension, up to
the combinatorial action of $\Sn_n$ on all faces.
In the case of six taxa, we get:
\[
%(169, 1032, 2626, 3489, 2492, 982, 182, 20).
(2, 20,  182,  982,  2492, 3489, 2626, 1032, 169).
\]

We illustrate these constructions and their properties in the case of
four taxa.

\begin{example}\label{ex:triangleN=4}
  \citep{BMEPolytope} 
  Fix $n=4$.
 The points $\ww^t$ are:
\[
\ww^{((1,2),(3,4))}=\frac{1}{4}[2,1,1,1,1,2] \; ; \; 
\ww^{((1,3),(2,4))}=\frac{1}{4}[1,2,1,1,2,1] \; ; \; \]
\[
\ww^{((1,4)),(2,3))}=\frac{1}{4}[1,1,2,2,1,1] \; ; \; 
\ww^{\text{star}(4)}=\frac{1}{3}[1,1,1,1,1,1] \; ; \; 
\]
  The BME polytope is a triangle in $\RR^6$ with vertices
  $\ww^{((1,2),(3,4))}, \ww^{((1,3),(2,4))}$ and $\ww^{((1,4),(2,3))}$. It
  spans the 2-dimensional space $\{(x_{12}, x_{13}, x_{14},
  x_{23}, x_{24}, x_{34}) \in \RR^6 \colon $
$  x_{12}+x_{13}+x_{14} =  x_{12}+x_{23}+x_{24}=
  x_{13}+x_{23}+x_{34}= x_{14} + x_{24}+x_{34} =1 \}$.
\end{example}

%We now come back to the BME fan. 
The \emph{lineality space} of a fan is defined as the maximal linear
space contained in all cones of the fan. 
If this space is just the
origin, we say that the fan is \emph{pointed}.
 In the case of the BME fan, this linear subspace is $n$-dimensional
 with basis given by the $n$ shift vectors $\shift_{\mathbf{a}}$
 corresponding to the $n$ leaves.
 Since the lineality space lies in all cones of the fan, we can
mod out by this subspace (for example, by taking a projection to its
orthogonal complement) and reduce our study to the case of pointed
complete polyhedral fans in $\RR^{\binom{n}{2}-n}$.
  We illustrate the
construction of the BME fan and the associated pointed fan on four
taxa.

\begin{example}
Let $n=4$.
 We mod out by the lineality space $L =\big(
  \shift_{\mathbf{1}}, \shift_{\mathbf{2}}, \shift_{\mathbf{3}}, \shift_{\mathbf{4}}\big)$
 via the \emph{canonical} projection map $p\colon \RR^{\binom{n}{2}} \to
 L^{\perp} \simeq \RR^{\binom{n}{2}-n}$ to the orthogonal complement
 of the subspace $L$ given by the matrix
\[
\left(
\begin{array}{rrrrrr}
0 & 1 & -1 & -1 & 1 & 0\\
1 & 0 & -1 & -1 & 0 & 1 \end{array}
\right).
\]
%Here, $L^\perp=\big( (1,-2,-1,-1,2,-1), (2,-1,-1,-1,-1,2)\big)$
We apply this projection to the BME fan, and we get a fan in $\RR^2$,
which we can plot.
 Alternatively, we project the BME polytope into 2-space and we take the inner normal fan of the resulting polytope.

  From Example~\ref{ex:triangleN=4} we know that the BME polytope is  the triangle with vertices corresponding to the three quartet trees  $((1,2),(3,4))$, $((1,3),(2,4))$ and $((1,4),(2,3))$. 
The projection $p$ maps this triangle to the triangle with vertices $(-2,4), (4,0)$ and $(-2,-2)$. Its inner normal fan consists of the rays spanned by  $r_1=(1,0)$, $r_2=(-1,-1)$ and $r_3=(0,1)$, plus the origin. 
 Figure~\ref{fig:BMECriteriaN=4} shows the quartets corresponding to the relative interior of each chamber.
  \begin{figure}[hts]
    \centering
    \includegraphics[scale=0.35]{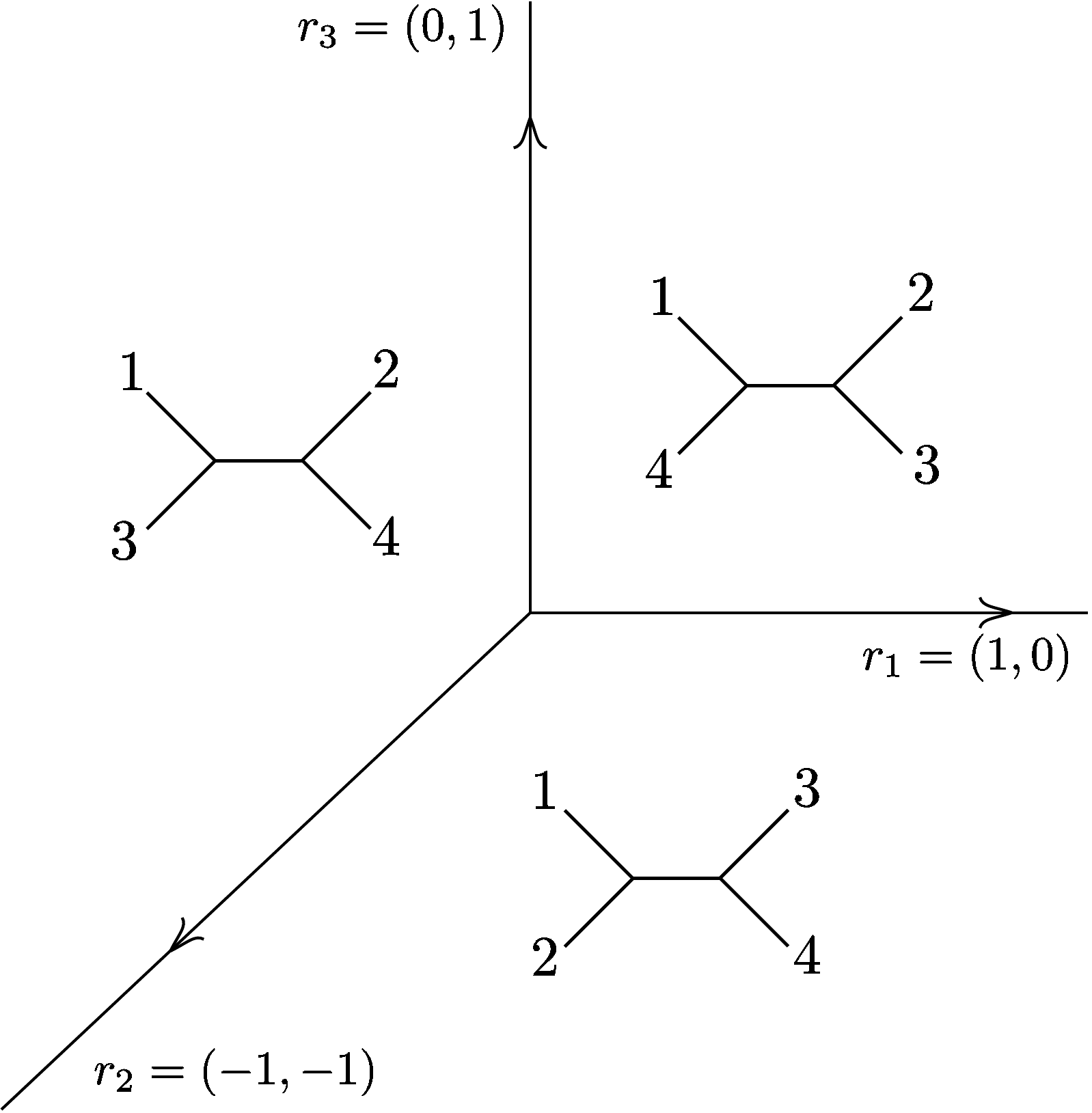}
    \caption{
    Quartets minimizing the BME criterion for each dissimilarity map on four taxa.
      }
    \label{fig:BMECriteriaN=4}
  \end{figure}
\end{example}

%%%%%%%%%%%%%%%%%%%%%%%%%%%%%%%%%%%%%%%%%%%%%%%%%%%%%%%%%%%%%%%%%%%%%%%%%%%%%
%%%%%%%%%%%%%%%%%%%%%%%%%%%%%%%%%%%%%%%%%%%%%%%%%%%%%%%%%%%%%%%%%%%%%%%%%%%%%

\section{Behavior of BME under the addition of an extra taxon}
\label{sec:behav-under-addit}

The purpose of this section is to investigate the relationship between lower and upper trees for arbitrary $D$.
Section~\ref{sec:funny-trees} shows that for any $D$ there exists a lifting such that the upper tree is as different as possible from the lower tree in terms of splits.
Section~\ref{sec:restrict-top} provides a counterpoint by
demonstrating that certain combinations of lower and upper trees are
not possible, i.e.\ that a rogue taxon cannot affect a BME tree in
arbitrary ways.

\begin{notation}
  Throughout the remainder of the paper, we label our taxa by $[n]=\{1, \ldots, n\}$.  
We write $\RR_+$ for the set of non-negative reals.
\end{notation}

\subsection{A theorem demonstrating the existence of unusual upper trees}
\label{sec:funny-trees}

We show that every lower tree has an upper tree whose restriction to the lower taxa is maximally different from it in terms of the \emph{Robinson-Foulds metric} $\RF$ on tree topologies, although perhaps not in terms of \emph{quartet distance}.
The $\RF$ metric on phylogenetic trees is defined in terms of bipartitions in the tree, also called ``splits.'' 
A split in a phylogenetic tree is simply the
bipartition of the taxa induced by cutting that edge.
  For example,
the split $\{1,2\}, \{3,4\}$ is induced by cutting the internal
edge of the quartet $((1,2),(3,4))$.
  Let $\splits(t)$ denote the set
of splits of tree $t$; the distance $\RF(s,t)$ is simply one half the size of the
symmetric difference of $\splits(t)$ and $\splits(s)$ \citep{robinsonFouldsComparison81}.

The quartet distance is analogous to the Robinson-Foulds distance but with the role of splits replaced by that of quartets (induced subtrees of size four) contained in a tree. 
The naive algorithm for computation is $O(n^4)$, although it be computed in $O(n^2)$ via a simple algorithm \citep{bryantEaQuartetDist00} and in $O(n \log n)$ via a more complex algorithm \citep{brodalEaQuartetDistNlogN03}.
In this paper, $s$ will have one more taxon than $t$; we accommodate this difference for the Robinson-Foulds and quartet distances by simply taking the induced tree on $s$ given by the set of lower taxa.

\begin{theorem}\label{thm:funny}
  Let $D$ be a dissimilarity map on $n$ taxa with BME tree $t$.
  There
  exists a lifting $\tilde{D}$ whose upper tree $s$ maximizes $\RF(s,t)$ among all trees on $n$ taxa.
\end{theorem}

This theorem will follow easily from the following lemma.

\begin{lemma}
  Given an ordering of $n$ taxa $z_1,\dots,z_n$ and any distance
  matrix $D$ on taxa $\{z_i\colon 1\leq i \leq n\}$, there exists a lifting
  $\tilde{D}$ such that the BME tree for $\tilde{D}$ restricted to
  $z_1,\dots,z_n$ is the caterpillar tree
  $(z_1,(z_2,\dots,(z_{n-1},z_n)\dots)$.
  \label{lem:funny}
\end{lemma}

\begin{proof}
  Pick arbitrary numbers $1 < \alpha_1 < \dots < \alpha_n$.
  Let $y$
  denote the extra ``rogue'' taxon. We construct a family of liftings
  $\tilde{D}^c$ as an exponential function for a given base number
  $c$.
  Set $\tilde{D}^c(y, z_i) = c^{\alpha_i}$.

  We write the BME length as
  \[ 
  \lambda(s,\tilde{D}^c) = \sum_{1\leq i<j\leq n} \ww^s_{i,j} D_{i,j} + \sum_{1\leq i\leq n} \ww^s_{i,n+1} c^{\alpha_i}.
  \] 
  % The first summation term as a function of $s$ is bounded above and below by positive constants, simply because there are a finite number of topologies on $n$ taxa.
  As $c$ goes to infinity, the dominant term in the summation becomes $\ww^s_{n,n+1} c^{\alpha_n}$. 
  For $c$ greater than some $c_n$, the BME tree must be a caterpillar tree with $y$ as far as possible from $z_n$. 
  Indeed, any other topology would have a smaller coefficient for $c^{\alpha_n}$.
  We can repeat the same argument replacing $n-1$ for $n$, finding a $c_{n-1}$ such that for $c \ge c_{n-1}$ the BME tree must be a caterpillar tree with $y$ as far as possible from the subtree $(z_{n-1},z_n)$.
  Continue in this way until a large enough lower bound on $c$ is found such that the described caterpillar tree is the BME tree for $\tilde{D}^c$.
\end{proof}

With this lemma, all that is needed to prove Theorem~\ref{thm:funny}
is to show that there exists a caterpillar tree $s$ such that the
restriction of the caterpillar to the original taxa has maximal
$\RF(s,t)$.

\begin{proof}[Proof of Theorem~\ref{thm:funny}]
  Color the taxa of $t$ with black and white colors as follows: for
  every cherry (two taxon subtree) of $t$, color one taxon white and
  the other black, and color the remaining taxa arbitrarily.
  Now
  order the taxa with all of the black taxa first and all of the white
  taxa second.
  The caterpillar tree from Lemma~\ref{lem:funny} using
  this ordering will have the required maximal $\RF$.
\end{proof}

\begin{remark}
  \label{rmk:quartet}
  The extension of Theorem~\ref{thm:funny} to quartet distances does not hold for more than seven taxa.
  Indeed, let $t$ be $(1,((((((2,3),4),5),6),7),8))$. 
  The maximally quartet-distant trees on 8 taxa (of quartet distance 61) are the following non-caterpillars:
  \[
  \begin{aligned}
    & (1,(2,((((3,8),5),(4,7)),6))) \\
    & (1,(2,((((3,8),6),(4,7)),5))) \\
    & (1,(((((2,8),5),(4,7)),6),3)) \\
    & (1,(((((2,8),6),(4,7)),5),3)). \\
  \end{aligned}
  \]
  These trees were found by our code and distances were confirmed with the \texttt{qdist} program of \citet{qdist}.
\end{remark}
  One could perform a similar analysis for the path distance metric of \citet{steelPennyTreeMetrics93}, although we have not done so.
 
\subsection{A theorem restricting topology of upper trees}
\label{sec:restrict-top}
The previous section shows that the lower and upper trees can be quite different.
It is natural then to ask about the collection of possible upper trees for a given lower tree.
That is, if we have a dissimilarity map $D$ on
$n$ taxa with BME tree $t$, what are the possible BME trees $s$ for
liftings of $D$?  This question narrows the potential effect of rogue
taxa.

We first gain intuition by investigating the case of four taxa.
This setting is simple, as there is only one trivalent tree topology on five taxa (up to relabeling of its leaves).
% Therefore, we wish to study the restriction of the projection map ${\pi_{BL}}{\restr{s\restr{[n]}}}$ to BME trees $s$ associated to lifting of the quartet $t=((1,2),(3,4))$ as well as the polyhedral cone $X_s(t)$.  

Using \polymake\ one can show that all but two tree topologies can be realized as upper trees for a lower quartet. 
The two trees not above $((1,2),(3,4))$ are shown in Figure~\ref{fig:notAbovet}.
\begin{figure}[ht]
  \centering
  \includegraphics[height=2cm]{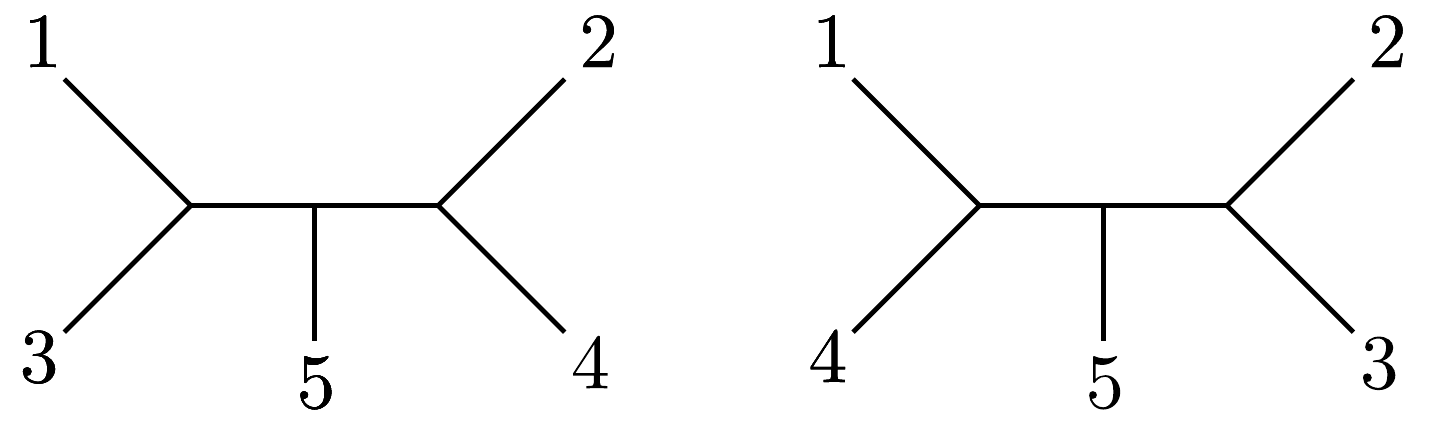}
  \caption{The trees that do not sit above $((1,2),(3,4))$ for any lifting of a dissimilarity map $D$ with BME tree $((1,2),(3,4))$.}
 \label{fig:notAbovet}
\end{figure}

This example can be established analytically and generalized to the
case of more taxa by replacing the leaves $1$ through $4$ with rooted
subtrees $a$ through $d$.
  In particular, we show that we can never
obtain a tree where pairs of subtrees are exchanged ``over'' the extra taxon.

Let $y$ denote the new leaf to be attached. 
The original tree $t$ is the tree $((a,b),(c,d))$. 
Call $s$ the tree $((a,c),(b,d))$ as in Figure~\ref{fig:clustersSAndT}.
\begin{figure}[ht]
  \centering
 \includegraphics[width=12cm]{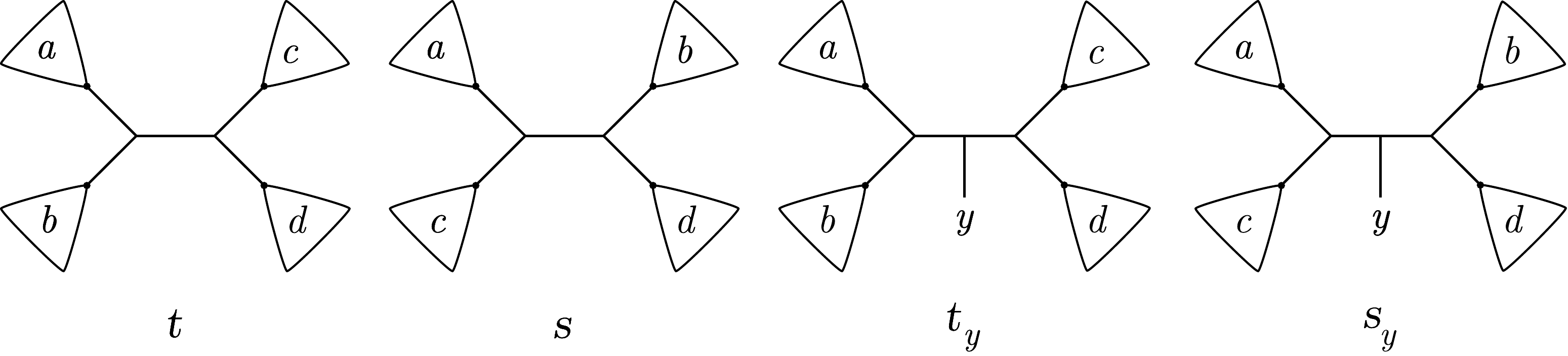}
 \caption{The trees $t$, $s$, $t_y$ and $s_y$.}
  \label{fig:clustersSAndT}
\end{figure}

\begin{theorem}\label{thm:NeverSx}
  Let $D$ be a dissimilarity map such the BME score of $t=((a,b),(c,d))$  is strictly less than that of $s=((a,c),(b,d))$ (Figure~\ref{fig:clustersSAndT}).
  Then the BME score of $t_y:=((a,b),y,(c,d))$ is strictly less than that of $s_y:=((a,c),y,(b,d))$ for any lifting $\tilde{D}$ of $D$.
  Consequently, if $t$ is the BME tree for $D$, then $s_y$ \emph{cannot} be a BME tree for any lifting $\tilde{D}$.
\end{theorem}
\begin{proof}
  We denote with sans serif font the elements in each subtree, so $\ssa$ denotes a leaf in subtree $a$, etc. 
  For simplicity we abbreviate $\ww^t$ by $\ww$.
  By definition, we get
  \[
  \ww^{s_y}_{\ssa\ssb}=\ww_{\ssa\ssb}/4 \; ; \; \ww^{s_y}_{\ssa\ssc}=2
  \ww_{\ssa\ssc} \; ; \; \ww^{s_y}_{\ssa\ssd}=\ww_{\ssa\ssd}/2 \; ; \;
  \ww^{s_y}_{\ssb\ssc}=  \ww_{\ssb\ssc}/2 \; ; \; \ww^{s_y}_{\ssb\ssd}=2
  \ww_{\ssb\ssd} \; ; \; \ww^{s_y}_{\ssc\ssd}= \ww_{\ssc\ssd}/4;
  \]
  \[
  \ww^{t_y}_{\ssa\ssb}=\ww_{\ssa\ssb} \; ; \; \ww^{t_y}_{\ssa\ssc}=
  \ww_{\ssa\ssc}/2 \; ; \; \ww^{t_y}_{\ssa\ssd}= \ww_{\ssa\ssd}/2 \; ; \;
  \ww^{t_y}_{\ssb\ssc}= \ww_{\ssb\ssc}/2 \; ; \; \ww^{t_y}_{\ssb\ssd}=
 \ww_{\ssb\ssd}/2 \; ; \; \ww^{t_y}_{\ssc\ssd}= \ww_{\ssc\ssd}.
  \]
  Similarly,
  \[
  \ww^{s}_{\ssa\ssb}=\ww_{\ssa\ssb}/2 \; ; \; 
  \ww^{s}_{\ssa\ssc}=2 \ww_{\ssa\ssc} \; ; \; 
  \ww^{s}_{\ssa\ssd}=\ww_{\ssa\ssd} \; ; \;
  \ww^{s}_{\ssb\ssc}=  \ww_{\ssb\ssc} \; ; \; 
  \ww^{s}_{\ssb\ssd}=2 \ww_{\ssb\ssd} \; ; \; 
  \ww^{s}_{\ssc\ssd}= \ww_{\ssc\ssd}/2.
  \]
  Since we are interested in the difference between the two scores, we
  do not compute the weights w.r.t.\ leaf $y$ nor weights within a
  cluster, since both trees have the same weight in these two cases.
  Then for any given lifting $\tilde{D}$ we have by subtraction
  \[
  \begin{split}
    \lambda(s_y, \tilde{D})-\lambda(t_y, \tilde{D})= 
%\langle
%    \tilde{D}, \ww^{s_y}\rangle -\langle \tilde{D}, \ww^{t_y}\rangle =
%  \sum_{a,b} d_{ab} (-3) \ww_{ab}/4
%     +\sum_{a,c} d_{ac} \,3\,\ww_{ac}/2 + \sum_{b,d} d_{bd}
% \,{3}\, \ww_{bd}/2 + \\\sum_{c,d} d_{cd} \ww_{cd} =
%    {3}/2\big( \langle D, \ww^{s}\rangle -\langle D,
 %   \ww\rangle\big)=
3/2\big(\lambda(s,D)-\lambda(t,D)\big).
  \end{split}\]
  The term on the right-hand side is positive by hypothesis.  
\end{proof}

\section{Liftings of tree metrics}
\label{sec:new-approach}

In the previous section, we analyzed the relationship between the
lower and upper trees for liftings of a general dissimilarity map $D$.
For a practicing phylogeneticist, however, this provides limited
useful information.
  Indeed, the basic assumption of phylogenetic
inference is that the data evolves in a primarily tree-like
manner.
 Namely, in distance-based inference, the assumption is that
the given dissimilarity map is ``close'' to a tree metric.
  In the
rogue setting, we are interested in $n$ taxa which evolve in a
tree-like manner and one, the rogue, that does not.

In this section we formalize these notions by assuming that $D$ is a
tree metric with respect to the tree topology $t$.
  By the consistency
of BME inference, the lower tree will be $t$.
  With this assumption,
our primary interest will be in understanding how the upper tree can
differ from $t$ in the sorts of situations more likely to be
encountered in phylogenetics.
  Although Theorem~\ref{thm:funny}
provides an interesting theoretical result in this vein, the required
lifting is quite unlikely to appear in data.
  By reformulating the
problem below directly in terms of the branch lengths of the tree
metric, we are able to obtain more precise and relevant information
about the action of rogue taxa.

\subsection{Preliminaries}
\label{sec:preliminaries}
% Throughout this paper, our taxa will be labeled by $[n]=\{1, \ldots, n\}$. 
% We denote the set $\RR_{\geq 0}$ by $\RR_+$.
\begin{notation}
  Given a positive integer $n$, we define $\cD_n$ to be the cone of
  dissimilarity maps on $n$ taxa.
  We identify $\cD_n$ with
  $\RR_+^{\binom{n}{2}}$.  %We define $M_n\subset \cD_n$ to be the set of
  %\emph{metrics} on $n$ taxa. 
  % (i.b=e.\ positive-entry dissimilarity maps satisfying the triangle
  % inequality).
Similarly,   we define $\cT_n\subset \cD_n$ to be  the space of tree metrics on $n$ taxa.
  We omit the subscript $n$ whenever it is clear from the context.
  Finally, given a tree topology $t$, we denote by $\cT_t\subset \cT_n$ the set of tree metrics with underlying tree topology $t$.
\end{notation}
\begin{notation}
  Given a trivalent tree $t$, the BME cone $\cc_{\ww^t}$ associated to $t$
  will be denoted by $\cc_t$.
 Moreover, we call $\cc_t^+=\cc_t \cap
  \RR_+^{\binom{n}{2}}$ the \emph{positive BME cone} of $t$, also
  known as the BME cone of dissimilarity maps associated to $t$.
\end{notation}
  \begin{notation}
    In what follows, we write $\PP_n$ for the BME polytope on $n$
    taxa.
 If the number of taxa is understood, we omit the subscript.
  \end{notation}
 
Given a tree topology $t$ on $n$ taxa, let $\pi_{t}\colon \RR_+^{\binom{n}{2}}\to \RR^{2n-3}$ denote a map generalizing the branch length map for tree metrics as follows.
The coordinates of this map are indexed by the branches of the tree $t$, and each coordinate is a linear function on the metric cone whose value on tree metrics with topology $t$ is precisely the length of the corresponding edge.
Note that this linear function is not unique, and it is positive on tree metrics with topology $t$. 
An expression defining the coordinate $e$ of the map $\pi_t$ (that is, the branch length of $e$) can be obtained by the four-point condition equations~\citep[Theorem 2.36]{ASCB} characterizing the tree topology $t$. 
For example, let $t= ((1,2),(3,4))$, let $e_i$ be the edge
adjacent to leaf $i$, let $e$ be the internal edge, and let $b_{e_i},
b_{e}$ be their corresponding lengths. Then $\pi_{t}(D) := (b_{e_1}(D), b_{e_2}(D), b_{e_3}(D),  b_{e_4}(D), b_e(D))$, where $b_{e_1}(D)=(D_{31}-D_{32}+D_{12})/2$,  $b_{e_2}(D)=(D_{32}-D_{31}+D_{12})/2$,  $b_{e_3}(D)=(D_{23}-D_{24}+D_{34})/2$, $b_{e_4}(D)=(D_{24}-D_{23}+D_{34})/2$,  and $b_e(D)=(D_{13}+D_{24}-D_{12}-D_{34})/2$. 
The map $\pi_t$ has the property that it identifies the cone of tree metrics realizing $t$ with $\RR_+^{2n-3}$.
 
Our goal for this subsection is to understand the interplay between the
branch lengths of a tree metric $D\in \cT_t$ and the possible upper
trees one can obtain by lifting this metric. In particular, we wish to
characterize the branch lengths of lower trees admitting a prescribed
upper tree $s$. 
It is clear that if we start from a tree metric $D=d_t$ and its
corresponding branch length vector $\pi_{t}(D)$,  we can easily lift
$D$ to a tree metric $\tilde{D}$ whose underlying tree $s$ contains $t$ as a subtree. 
Hence, the union of the sets $\{\pi_{t}(D)\colon D \ \mathrm{s.t.} \ \exists\, \tilde{D}\in
\cc_s^+\}$ as $s$ varies among a possible upper BME trees equals the
set $\RR_+^{2n-3}$. 
We want to understand each one of these sets. In particular, we want
to answer the following challenge:
\begin{problem}\label{pb:partUpper}
  Given a tree topology $t$ on $n$ taxa and $s\in \cT_{n+1}$, describe
  the cone of dissimilarity maps on $n+1$ taxa whose BME tree equals
  $s$ and whose restriction to the first $n$ taxa is a tree metric of
  combinatorial type $t$.  
\end{problem}

For each upper tree $s$, the elements of the corresponding set in
Problem~\ref{pb:partUpper} can be thought of as vectors in $\RR_+^{3n-3}$,
where the first $2n-3$ entries
encode the branch lengths of the lower tree $t$ and the remaining ones
refer to distances to the new taxon. That is,
\begin{equation}
X_s(t):=\{ (\pi_{t}(D), \tilde{D}_{1,n+1}, \ldots, \tilde{D}_{n,n+1}) \colon D \in
\cT_t, \tilde{D} \in \cc_s^+\}.
\end{equation}
By construction, these sets are polyhedral cones and they partition
the set $\RR_+^{3n-3}$:
\begin{proposition}
  $X_s(t)$ is a rational (possibly empty) polyhedral cone for every $s$ and $t$.
 It is described by two types of \emph{homogeneous linear} constraints:
\begin{itemize}
\item %\textbf{Type 1:} 
all entries $\tilde{D}_{ij}\geq 0$ and $\pi_{t}(D)\geq 0$.
\item %\textbf{Type 2:}
 inequalities describing $\cc_s$: they correspond to the directions $\ww^s-\ww^u$ for all  trivalent trees $u$ on $n+1$ taxa, and all constants are zero. 
That is: $\langle \ww^s-\ww^u, \tilde{D}\rangle \geq 0$, for all trivalent trees $u$.
\end{itemize}
\end{proposition}
\begin{proof}
  $X_s(t)$ is a polyhedral cone because it is the image of the linear
  map $\tilde{D}\mapsto
(\pi_{t}(\tilde{D}\restr{[n]}), \tilde{D}_{1,n+1},
  \ldots, \tilde{D}_{n,n+1})$, where $\tilde{D}\in
  \cc_s^+\cap(\cT_t\times \RR^n_+)$. 
 The inequalities describing  $X_s(t)$ follow by construction. 
The entries of $\tilde{D}\restr{[n]}$ are expressed as linear
combinations of  the entries $\pi_{t}(\tilde{D}\restr{[n]})$.
The second group of inequalities include facet inequalities of the
cone $\cc_s$: whose directions are given by the edges containing vertex $\ww^s$.  
To  simplify the construction, we add the inequalities coming from
  differences between $\ww^s$ and all other vertices of $\PP$ and
  not only of vertices $\ww^u$ adjacent to $\ww^s$. 
Adding these inequalities makes no harm and it simplifies the problem by avoiding the computation of the edges adjacent to $\ww^s$, which can be hard if the number of taxa is too big.
%   $X_s(t)$ is a polyhedral cone because it is the image of the linear
%   map $\tilde{D}\mapsto
%   \big(((\pi_{BL}\restr{t}(\tilde{D}\restr{[n]}), \tilde{D}_{1,n+1},
%   \ldots, \tilde{D}_{n,n+1}), \pi_{BL}\restr{s_{|_{[n]}}
%   }(\tilde{D}\restr{[n]})\big)$, where $\tilde{D}\in
%   \cc_s^+\cap(\cT_t\times \RR^n_+)$.  The inequalities describing
%   $X_s(t)$ (i.e.\ its facets) follow by construction. The type 1
%   inequalities express the coordinates of
%   $\pi_{BL}\restr{s_{|_{{[n]}}}}(\tilde{D}\restr{[n]})$ as a linear
%   function on the coordinates of
%   $\pi_{BL}\restr{t}(\tilde{D}\restr{[n]})$, whereas the inequalities
%   of type 3 correspond to the facets of the cone $\cc_s$: its facets
%   have directions given by the edges containing vertex $\ww^s$.  To
%   simplify the construction, we add the inequalities coming from
%   differences between $\ww^s$ and all other vertices of $\PP$ and
%   not only of vertices $\ww^u$ adjacent to $\ww^s$. Adding these inequalities makes no harm and it simplifies the problem by avoiding the computation of the edges adjacent to $\ww^s$, which can be hard if the number of taxa is too big.
\end{proof}

\subsection{The reduced BME polytope}
\label{sec:reduced-bme-polytope}
We now present an equivalent approach to our lifting task in the setting of this section, i.e.\ when $D$ is a tree metric on $n$ taxa with (trivalent) tree $t$ and branch lengths ${\bf b}_e$.
  As shown below, all that is needed to study the
restricted BME problem is a change of order of summation followed by a
grouping of appropriate terms.
  This small modification reduces the
problem from having a quadratic number of free variables to a linear
number, as well as simplifying the constraints.
  After introducing the
reduced polytope, we show that it has dimension $2n-4$ by
characterizing its affine hull.

The set of edges of $t$ will be denoted by $E(t)$.
Pick any lifting $\tilde{D}$ of $D$, and any tree $s$ with $n+1$ leaves. 
The BME length of $s$ with respect to $\tilde{D}$ can be calculated as follows:
\[
\lambda (s,\tilde{D})= \langle \ww^s, \tilde{D}\rangle = \sum_{i,j\neq
  n+1} \ww^s_{ij} {D}_{i,j} + \sum_{i=1}^n 
\ww^s_{i,n+1} \tilde{D}_{i,n+1}.
\]
Now we simply substitute in the definition of the dissimilarity map $D$:
\[
D_{i,j} = \sum_{e \in t(i \tofrom j)} {\bf b}_e,
\]
where $e\in t(i\tofrom j)$ indicates that edge $e\in E(t)$ lies in the path between leaves $i$ and $j$ in tree $t$.  
Exchanging order of summation and regrouping,
\begin{equation}
  \langle  \ww^s, \tilde{D}\rangle = 
  \sum_{e\in E(t)} 
    \bigg(\sum_{\substack{i,j\neq n+1\\ e \in t(i\tofrom j)}}
    \ww^s_{ij}\bigg) \, {\bf b}_e
    + \sum_{i=1}^n  \ww^s_{i,n+1}  \tilde{D}_{i,n+1}
  \label{eq:7}
\end{equation}
which is again a simple inner product with a rational vector.
For a tree $s$ on $n+1$ taxa, define $(\vvv^s)_{\cdot}\in \RR^{3n-3}$ by
\begin{equation}
  \begin{cases}
    (\vvv^s)_e  = \hbox{\!} \sum\limits_{\substack{i,j\neq n+1\\ e \in t(i\tofrom j)}} \ww^s_{ij}, & e \hbox{ edge of lower tree} \\
    (\vvv^s)_i  = \hbox{ } \ww^s_{i,n+1}, & 1 \le i \le n.
  \end{cases}\label{eq:5}  
\end{equation}
Note that this definition depends on the fixed tree $t$, but we do not incorporate it to the notation, as we will typically be fixing a lower tree.

To find the BME tree for a tree metric $(t,\{{\bf b}_e\}_{e\in E(t)})$, we build a vector $\vvv^s \in \RR^{3n-3}$ for each tree $s\in \cT_{n+1}$.
Each vector has entries indexed by the $2n-3$ edges of $t$ and the $n$
distances $\{\tilde{D}_{i,n+1}\colon i=1, \ldots, n\}$.
Our goal is to find $s$ minimizing the quantity~\eqref{eq:7}.
As in the case of the BME problem, we build a polytope $\bb^t$ (here in $(3n-3)$-space) which is the convex hull of the points $\vvv^s$ and study its properties. 

\begin{definition}
Fix a tree $t$ on $n$ taxa and consider the points $(\vvv^s)_{e,i}$ as in~\eqref{eq:5}. 
The convex hull of these points  is called the  {``reduced BME polytope''}, and we denote it by $\bb^t$. 
It only depends on the combinatorial type of the tree $t$ and it is symmetric with respect of the group of symmetries of the tree $t$.
  The points $\{\vvv^s\colon s\in
  \cT_{n+1}\}$ are called ``reduced weights.''
  The \emph{inner} normal fan of $\bb^t$ is called the ``reduced fan.''
  Cones in  this fan are called {``reduced cones''} and their intersections with the positive orthant are be called ``positive reduced cones.''
\end{definition}

From the previous construction it is clear that the BME polytope and the reduced BME polytope are closely related. 
We now explain this connection.  
The linear map $\alpha_t \colon \RR^{\binom{n+1}{2}} \to \RR^{3n-3}$ assigning the reduced weight $\vvv^s$ to  the BME weight $\ww^s$ sends the polytope $\pp$ surjectively onto  the polytope $\bb^t$. 
That is, the reduced polytope is a linear \emph{projection} of the BME polytope. 
On the dual side, the dual of the linear map $\alpha_t$ will inject the dual space of the polytope $\bb^t$ into the dual space of the polytope $\pp$, and in this case the linear spaces of both polytopes are identified by the map $\alpha_t$ (Proposition~\ref{pr:EqnB^t}).
We refer the interested reader to \citep[Section
7.2,][]{zieglerLectConvexPoly95} for more information about
projections of polytopes.

\begin{example}\label{ex:4to5Reduced}
  We illustrate the previous construction in the case of liftings of the quartet tree $t=((1,2), (3,4))$, describing the reduced weights $\vvv^s$ for six trivalent trees $s$ in Table~\ref{tab:4to5taxa}. 
The remaining reduced weights can be obtained by relabelings of $s$ that respect the combinatorial type of $t$.
The table is organized as follows.
 The first five columns encode the branch lengths of the lower tree: ${\bf b}_0$ for the
internal edge of $t$, and ${\bf b}_i$ for the edge pendant to taxon $i$.
The rest, $x_1$ through $x_4$ are the four distances to the new taxon.
 The polytope $\bb^{((1,2),(3,4))}\subset \RR^9$ is four-dimensional, 
 has 14 vertices and $f$-vector $(14, 46, 52, 20)$. 
The vertices of $\pp_5$ corresponding to the trees $((1,3), (5,
(2,4)))$ and $((1,4), (5, (2, 3)))$ project to the same vertex of
$\bb^t$.
Among all 14 vertices, only 5 correspond to upper BME trees: the reduced weight corresponding to the tree $s=((2,5),(3, (1, 4)))$ and its five relabelings that fix $t$.
 The affine hull of $\bb^t$ has five defining linear equations $x_1+x_2+x_3+x_4=1$ and ${\bf b}_{i}+x_i$ for $i=1,2,3,4$. 
 Analogous equations will define the affine hull for all reduced BME polytopes, as we show in Proposition~\ref{pr:EqnB^t}.
\end{example}

  \begin{table}[h]
    \centering
   \begin{tabular}{|c||c|c|c|c|c||c|c|c|c|}
\hline     upper tree & ${\bf b}_{1}$ & ${\bf b}_{2}$ & ${\bf b}_{3}$ & ${\bf b}_{4}$ & ${\bf b}_{0}$ & $x_1$ & $x_2$ & $x_3$ & $x_4$\\
     \hline 
     $((1,2),(3, (4,5)))$ & 7/8 & 7/8 & 6/8& 4/8 & 6/8 &   1/8& 1/8 & 2/8 & 4/8\\
\hline
     $((1,2),(5, (3,4)))$ & 6/8 & 6/8 & 6/8& 6/8 & 4/8 &    2/8& 2/8 & 2/8 & 2/8\\
\hline
$((1,3),(2, (4,5)))$ & 7/8 & 6/8 & 7/8& 4/8 & 9/8 &    1/8& 2/8 & 1/8 & 4/8\\
\hline
     $((1,3),(5, (2,4)))$ & 6/8 & 6/8 & 6/8& 6/8 & 10/8 &    2/8& 2/8 & 2/8 & 2/8\\
\hline
     $((1,3),(4, (2,5)))$ & 7/8 & 4/8 & 7/8& 6/8 & 9/8 &    1/8& 4/8 & 1/8 & 2/8\\
\hline
     $((1,5),(2, (3,4)))$ & 4/8 & 6/8 & 7/8& 7/8 & 6/8 &    4/8& 2/8 & 1/8 & 1/8\\
\hline
        \end{tabular}

   \caption{
   Reduced weights for trivalent trees on five taxa, starting from the lower tree $t=((1,2), (3,4))$, up to symmetry of the lower tree $t$.  
   The column labels show the quantity for which the entry is the
   corresponding coefficient in the reduced weight vector: e.g. the first entry of the table shows that $7/8$ is the coefficient of ${\bf b}_{1}$ for topology $((1,2),(3,(4,5)))$.
   }
   \label{tab:4to5taxa}
 \end{table}

One can compute the dimension, number of vertices, and $f$-vector of the reduced polytope $\bb^t$ as we did in the case of the BME polytope. 
We can also study the behavior of the vertices of the BME polytope
under the projection map, and see how many of its vertices collapse to
a single vertex in $\bb^t$, how many lie in the interior and how many
lie in proper faces of positive dimension. 
We now show that the reduced polytope has dimension $2n-4$ by characterizing its affine hull.
First we state a technical lemma.
 Questions involving vertices and their behavior under the projection map will be deferred to the next section.
\begin{lemma}\label{lm:B^tEquations}
  Given a tree $t$ on $n$ taxa, let $\ww$ denote the BME weight for $t$.
  Then
  \[ 
\sum_{j\neq i} \ww_{ij}=1 \qquad \forall \ 1\leq i\leq n.
  \]
\end{lemma}
\begin{proof}
If a non-backtracking random walk starts at $i$, then $w_{ij}$ is the probability of that walk ending at $j$.
\end{proof}

\begin{proposition}\label{pr:EqnB^t}The affine hull of $\mathcal{B}^t
  $ is characterized by $n+1$ linearly independent linear
  equations.
 More precisely, they are
 given by $Ax=\bf{1}\in \RR^{n+1}$, where
\[
A:=\left(
  \begin{array}{c|c|c}
I_n & \bf{0}& I_n\\
\hline
    \bf{0} & \bf{0} & \bf{1}
  \end{array}
\right)\in \ZZ^{(n+1)\times(3n-3)},
\]
and the columns of $A$ and points in $\RR^{3n-3}$ are labeled by
partitioning the coordinates as $({\bf b}_{e_1}, \ldots, {\bf b}_{e_n}\,\mid \, {\bf b}_e \colon
e\text{ interior edges of }t\, \mid \,\tilde{D}_{1,n+1}, \ldots,
\tilde{D}_{n,n+1})$.
 Here, $e_i$ denotes the edge pendant to the leaf $i$
in tree $t$.
 In particular, $\dim \mathcal{B}^t=2n-4$, and the $(n+1)$-dimensional lineality space of the reduced fan coincides with the row span of $A$. 
\end{proposition}
\begin{proof}
First, we rewrite the equations in terms of the coordinates of reduced weights then apply Lemma~\ref{lm:B^tEquations}.
 Fix an upper tree $s$ and write
$\vvv$ and $\omega$  for $\vvv^s$ and $\omega^s$ respectively.
The following equalities hold:
\begin{equation*}
  \begin{aligned}
    \sum_{j=1}^n \vvv_{j} &= \sum_{j\neq n+1} \ww_{i\,n+1}=1 \\
    \vvv_{e_i} + \vvv_{i} &= \sum_{j\neq i}%{\substack{j=1 \\j\neq i}}^{n+1}
    \ww_{ij}=1 \qquad \forall \ 1\leq i\leq n.
  \end{aligned}
\end{equation*}
These are precisely the linear equations described by matrix $A$. 

We now prove that these equations characterize the space.
 To simplify notation, let $\psi$ be the {surjective} map $\psi(p)=
(\pi_{t}(p_{|_{[n]}}), p_{1, n+1}, \ldots, p_{n, n+1})$ for any
lifting $p$ of a tree metric with tree $t$.
  We proceed by dimensionality arguments.
  We know that $rk(A)=n+1$, so $\dim \cB^t\leq 3n-3-(n+1)=2n-4$.
 Our goal is to show that equality holds. 
 It will suffice to show that the dimension of the lineality space of the ``reduced
fan'' equals $n+1$.

By construction, the shift vectors $\{\shift_{a}\colon 1\leq a \leq n+1\}$
represent tree metrics associated to a degeneration of the trivalent
tree $t$ with two nodes and one edge: a leaf labeled $a$ and the other leaf
labeled by the set $\{1, \ldots, \widehat{a}, \ldots, n+1\}$.
 Hence, these tree
metrics can be expressed as points $\tilde{\shift}_a=\psi(\shift_a)$ in
$\RR^{3n-3}$ and they generate an $(n+1)$-dimensional vector space.
 These
points are precisely the rows of $A$ as described in the statement.
 Hence, it suffices to show that these vectors span the lineality
space of the ``reduced fan''.

Fix any trivalent tree $s_0$ on $n+1$ taxa. Given $p\in \RR^{3n-3}$ in
the lineality space of the reduced fan, by definition we have $\langle p,
\vvv^s\rangle=\langle p, \vvv^{s_0}\rangle$ for all trees $s$.
 By construction, $p$ lies in the image of $\psi$, so fix $q$ with
$p=\psi(q)% \widetilde{\pi_{t}(q)}
$.
 Thus, $  \langle {q}, \ww^s\rangle= \langle p, \vvv^s\rangle$ for all
$s$ by~\eqref{eq:7} and so $\langle {q}, \ww^s\rangle =\langle {q},
\ww^{s_0}\rangle$ for all $s$.
 By definition, we have that ${q}$ is in
the lineality space of the BME fan and so it is a linear combination
of the shift vectors.
 After applying the map $\psi$, the same holds
for $p$ and the vectors $\tilde{\shift}_a$, and the result follows.
\end{proof}

\subsection{Analysis of the reduced BME polytope}
\label{sec:analysis-reduced-bme}
In this section we focus on combinatorial properties of the reduced BME polytope and the behavior of the vertices of the BME polytope under the projection map $\alpha_t$, as $t$ varies along the set of combinatorial types of trees on $n$ taxa. 
In particular, we give a complete description of the vertices for up to six 
%eight taxa 
taxa (see Table~\ref{tb:ReducedUpTo8}). 
As we mentioned earlier, two tree topologies on $n+1$ taxa can give the same vertex in the polytope $\bb^t$ and vertices of the BME polytope can map to interior points in $\bb^t$ under the projection map.
As Example~\ref{ex:4to5Reduced} shows,  for four
taxa there exists a pair of tree topologies with the same associated
reduced weight, but all fourteen reduced weights are still vertices of
$\bb^t_4$.
  Similarly, in the case of five taxa, a \polymake\
computation shows that all 94 possible (out of 105) reduced weights $\{\vvv^s\colon s\in
\cT_6\}$ are vertices.
  This is no longer true for six taxa.%, seven and eight taxa. 

By construction, the polytope $\bb^t$ encodes an optimization problem where we restrict
our ambient space $\RR^{\binom{n+1}{2}}$ to the space of extensions of tree metrics with associated tree $t$. 
In terms of the BME fan, this means cutting out the fan with the $(2n-3)$-dimensional cone $\RR_+\cT_t\subset \RR^{\binom{n+1}{2}}$.
Note that by intersecting the BME chambers with this cone, we may get a cone with dimension less than $2n-3$.
Moreover, it could very well happen that this intersection is just the lineality
space $\RR(\alpha_t(\shift_a)\colon 1\leq a\leq n+1)$ of the cone.
This would imply that the point $\vvv^s$ lies in the interior of the polytope.
This is indeed what happens for six taxa, as we have found through computation: 
\begin{proposition}\label{pr:reducedForn6}
    Let $t=((1,2),(3,4),(5,6))$ be the snowflake tree. 
Then the  reduced polytope $\bb^t_6$ is generated by the $792$ reduced weights (out of the possible $945$ reduced trivalent points) and    it has $780$ vertices and $83\,227$ facets.
 The remaining twelve  reduced trivalent weights $\vvv^s$ that are not vertices of $\bb^t_6$  lie in the \emph{interior} of the polytope.
 They are associated to   pairs of trivalent trees with topologies:
\begin{verbatim}
(1,((((2,3),(4,6)),7),5)) (1,((((2,4),(3,6)),7),5))
(1,((((2,3),7),(4,6)),5)) (1,((((2,3),7),(4,5)),6))
(1,(((2,3),((4,6),7)),5)) (1,(((2,5),((4,6),7)),3))
(1,((((2,5),(3,6)),7),4)) (1,((((2,6),(3,5)),7),4))
(1,((((2,5),7),(3,6)),4)) (1,((((2,5),7),(4,6)),3))
(1,(((2,5),((3,6),7)),4)) (1,(((2,4),((3,6),7)),5))
(1,((((2,6),7),(3,5)),4)) (1,((((2,6),7),(4,5)),3))
(1,(((2,6),((3,5),7)),4)) (1,(((2,4),((3,5),7)),6))
(1,((((2,3),(4,5)),7),6)) (1,((((2,4),(3,5)),7),6))
(1,(((2,3),((4,5),7)),6)) (1,(((2,6),((4,5),7)),3))
(1,((((2,4),7),(3,6)),5)) (1,((((2,4),7),(3,5)),6))
(1,((((2,5),(4,6)),7),3)) (1,((((2,6),(4,5)),7),3))
\end{verbatim}
 Similarly if $t$ is the lower tree $(1,(((3,4),6),5),2)$ (the caterpillar tree), then the polytope $\bb_6^t$ has $804$ distinct reduced weights, $800$ vertices and $116\,701$ facets.
 In  this case, all four reduced trivalent weights $\vvv^s$ that are not vertices of $\bb^t_6$ lie in the interior.
 In this case, each point corresponds to a single topology and they are:
\begin{verbatim}
(1,((((2,(3,5)),7),4),6))
(1,((((2,6),3),7),(4,5)))
(1,((((2,(4,5)),7),3),6))
(1,((((2,6),4),7),(3,5)))
\end{verbatim}
\end{proposition}

From the previous examples, we see that in the case of four and five
taxa, all reduced points are vertices.
 And for six taxa, reduced
points are either vertices or \emph{interior} points (Proposition~\ref{pr:reducedForn6}).
 Thus, it is natural to ask if these are the only two possibilities:
\begin{question}
  For $n\geq 7$ and any tree $t\in \cT_n$, are all reduced trivalent points  either vertices or interior points of the  reduced polytope $\bb^t$?
\end{question}
\noindent We expect the answer to be positive, provided the projection map $\alpha_t$ is generic.

\medskip

We now switch gears and focus on the number of upper BME trees we can obtain from a lifting of a given tree metric with topology $t$. This study will highlight the behavior of ``rogue taxa.''
Equivalently, we want to know how many positive reduced cones $\cc^+_{s}(\bb^t)$ ($s$  trivalent tree on $n+1$ taxa) are non-empty.
We provide a complete answer for up to six taxa in Table~\ref{tb:ReducedUpTo8} below.

The next natural question to ask is what are the asymptotics (or provide an upper bound) of the number of such non-empty positive reduced cones.
As a first attempt, we give some insight about which topologies can be ruled out for upper BME trees.
 In other words, which are the \emph{blocking} topologies for upper trees.
\begin{definition}
Fix $t\!\in\!\cT_n$ and let  $\vvv^s$ be the reduced weight for a trivalent tree $s\!\in\! \cT_{n+1}$.
 We define a \emph{partial order} on the set $\{\vvv^s\colon s\in \cT_{n+1}\}$ as follows: $\vvv^s\succ \vvv^{s'}$ if and only if $(\vvv^s)_l \leq (\vvv^{s'})_l$ for all $1\leq l \leq 3n-3$.
 We say $s$ \emph{blocks} $s'$ if $\vvv^s\succ \vvv^{s'}$.
\end{definition}

\begin{lemma}
  Let $t\in \cT_n$, and $s,s'\in \cT_{n+1}$ be such that $s$ blocks
  $s'$.
 Then, $s'$ cannot be a BME tree for any lifting $\tilde{D}$ of
  $D\in \cc_t^+$. 
\end{lemma}
\begin{proof}
  It suffices to show that for any $\tilde{D}$, $\lambda(s,\tilde{D})  \leq \lambda(s',\tilde{D})$, and this follows because $\tilde{D}$ has non-negative entries.
\end{proof}
We illustrate with examples on five taxa.

\begin{example}
  Let $t=(1,((3,4),5),2)$.
 Out of all possible 94 vertices in $\bb^t$, there are 19 reduced vertices that are blocked by other vertices, out of 20 empty positive reduced cones.
 The blocking  relation is described in Figure~\ref{fig:domination6} and it gives 26 blocking upper tree topologies.
 We simplify  the picture by reducing the relation modulo relabeling of all
  leaves involved in each chain and that fix the lower tree $t$.
  
\begin{figure}[ht]
    \centering
    \includegraphics[width=10cm]{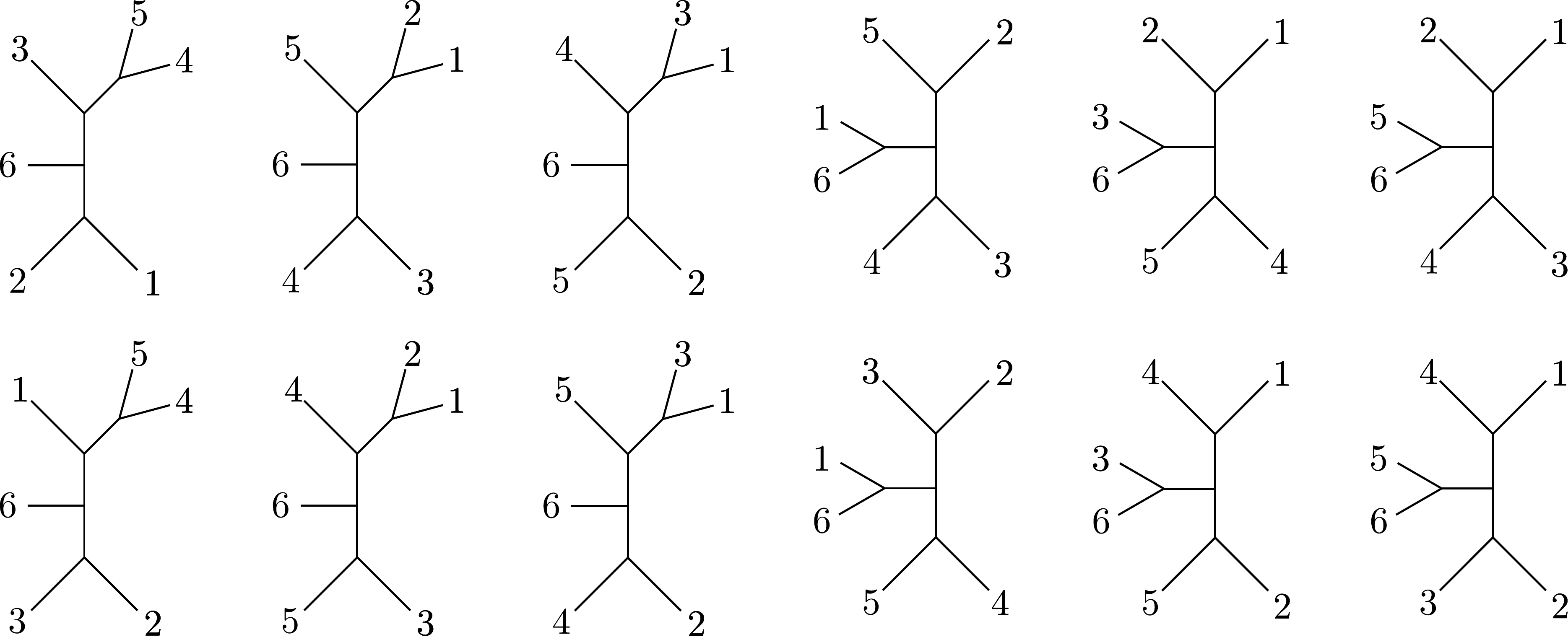}
    \caption{
    The blocking relations (up to symmetry) for trees on six taxa. 
    Pairs of trees in a column are a single blocking relation, with the  tree in the second row blocking the corresponding tree in the first row.
    Note that these blocking relations do not come from Theorem~\ref{thm:NeverSx}.
    }
    \label{fig:domination6}
  \end{figure}
In particular we see that out of the 94 possible BME reduced vertices
for $t$, we can rule out 19 of these vertices for upper trees by ``blocking'' relations.
\end{example}

Unfortunately, this partial order set is not a sufficient criterion to
determine if a tree on $n+1$ taxa can be an upper tree or not.
 In particular, it cannot explain the obstruction to exchange subtrees ``over'' the new pendant edge (Theorem~\ref{thm:NeverSx}), except in the case of quartet trees. However, understanding the blocking relation can give an upper bound for the asymptotics of the upper BME trees.

We end this section with a table discribing the relation between the
BME and reduced BME polytopes for up to six 
taxa.
In the case of six 
taxa, we have two combinatorial types of lower trees and each one will label a row in our table.
The row starting with ``6a'' indicates the caterpillar tree on six taxa, whereas ``6b'' refers to the snowflake tree (see Proposition~\ref{pr:reducedForn6}).
\begin{table}[ht]
  \centering
  \begin{tabular}{|c||c|c|
| c|c||c |c|}\hline
n 
&\multicolumn{2}{|c||}{ {dim.}} 
&  \multicolumn{2}{|c||}{$\#$ vertices} & {$\#$ void upper } & {$f$-vector of reduced BME}\\
\cline{2-5}  
& \text{BME} & \text{red.} &  \text{BME} & \text{red.} & trees for $t$&  positive cones\\
\hline
3 & 2 & 2 & 3 & {3} & 0 & (0,0,3)\\
\hline
4 & 5 & 4 & 15 & {14}& 2& $(1, 0,0,0,13)$\\
\hline
5 & 9 & 6 & 105 & {94} &  20 & $(16, 1, 6, 0, 0, 0, 71)$\\
\hline
% (32, 0, 38, 3, 14, 0, 0, 0, 15)\\
6a &  14 & 8 & 945 & 800 & 208 &  $(160, 32, 98, 10, 39, 0, 0, 461)$ \\
6b & 14 & 8 & 945 & 780  & 154 & $(123, 0, 144,9, 39, 0, 0, 0, 465) $ \\
%\hline
%7 & 20 & 10 & 10395 & ( ?, ?) & ?&(?, ?)\\
%\hline
%8 & 27 & 12 & 135135& (?, ?, ?, ?, ?) & ? & (?, ?, ?, ?, ?)\\
\hline
\end{tabular}
\caption{A comparison between the BME and reduced BME polytopes for up
  to six %eight 
taxa. 
In the case of six %or more 
taxa, we have more than one combinatorial type for the lower tree $t$. Each vector in the last column gives the number of reduced BME positive cones classified by dimension, starting from dimension $n+1$ and up to dimension $3n-3$.
 The lowest dimensional ones correspond to reduced weights of forbidden upper BME trees, since they lie in the linear space spanned by the shift vectors.
The discrepancy between the first entry of these vectors and the entry of the column indicating the number of voided upper trees reflects that several of these void trees have equal reduced weights.
}
\label{tb:ReducedUpTo8}
\end{table}

We conclude with an interesting computationally challenging question:
\begin{question}
What are the \emph{asymptotics}  of the number of vertices of the
$\bb^t$  and of the number of upper BME trees and upper BME reduced trees for different combinatorial types of lower trees $t$?  
\end{question}

\subsection{The rogue taxon effect for four taxa}
\label{sec:rogue-taxon-effect}
The extremal rays of each reduced cone can be interpreted to give
precise information on the rogue taxon effect.
  In this section, we
explore the reduced polyhedral cone associated to the lower tree
$((1,2),(3,4))$ and the upper tree $(((1,5),3),(2,4))$.
% note that this is tree 05, aka (0,(((1,3),2),4))
  Up to symmetry, this is the only lower/upper combination for this number of taxa such that the new taxon has ``rogue'' behavior.
  By understanding the extremal rays of the polyhedral cone, we establish Propositions~\ref{prop:rogueFunnyBL} and \ref{prop:rogueThreeTimes}.

\begin{table}[ht]
  \centering
  \begin{tabular}{|l||llll|l|llll|}
\hline      & ${\bf b}_1$ & ${\bf b}_2$ & ${\bf b}_3$ & ${\bf b}_4$ & ${\bf b}_0$ & $x_1$ & $x_2$ & $x_3$ & $x_4$ \\
     \hline
     \hline
$c_1$       &  4  &  0  &  3  &  3  &  1  &  0  &  0  &  0  &  0  \\
$c_2$       &  3  &  0  &  3  &  0  &  1  &  0  &  0  &  0  &  0  \\
$c_3$       &  1  &  0  &  0  &  0  &  1  &  0  &  3  &  0  &  0  \\
$c_4$       &  0  &  0  &  0  &  0  &  1  &  0  &  4  &  1  &  1  \\
$c_5$       &  0  &  0  &  0  &  0  &  1  &  0  &  3  &  0  &  3  \\
\hline
$e_1$       &  1  &  1  &  1  &  1  &  0  &  0  &  0  &  0  &  0  \\
$e_2$       &  1  &  1  &  1  &  0  &  0  &  0  &  0  &  0  &  0  \\
$e_3$       &  1  &  0  &  1  &  1  &  0  &  0  &  0  &  0  &  0  \\
$e_4$       &  1  &  0  &  1  &  0  &  0  &  0  &  0  &  0  &  0  \\
$e_5$       &  1  &  0  &  0  &  0  &  0  &  0  &  0  &  0  &  0  \\
$f_1$       &  0  &  0  &  0  &  0  &  0  &  1  &  1  &  1  &  1  \\
$f_2$       &  0  &  0  &  0  &  0  &  0  &  0  &  1  &  1  &  1  \\
$f_3$       &  0  &  0  &  0  &  0  &  0  &  0  &  1  &  0  &  0  \\
$f_4$       &  0  &  0  &  0  &  0  &  0  &  0  &  0  &  0  &  1  \\
$\shift_1$  &  1  &  0  &  0  &  0  &  0  &  1  &  0  &  0  &  0  \\
$\shift_2$  &  0  &  1  &  0  &  0  &  0  &  0  &  1  &  0  &  0  \\
$\shift_3$  &  0  &  0  &  1  &  0  &  0  &  0  &  0  &  1  &  0  \\
$\shift_4$  &  0  &  0  &  0  &  1  &  0  &  0  &  0  &  0  &  1  \\

\hline
  \end{tabular}

  \caption{
  The extremal rays of the polyhedral cone $X_s(t)$ for four lower taxa for $t=((1,2),(3,4))$ and $s=(((1,5),3),(2,4))$.
  The rows represent the rays.
  Labeling conventions for rows and columns are described in the text.
  }
  \label{tab:coneverts4}
\end{table}

Table~\ref{tab:coneverts4} gives the extremal rays of the cone
$X_s(t)$. We follow the notation of Example~\ref{ex:4to5Reduced} to
label the columns.
% The first five are the branch lengths of the lower tree: ${\bf b}_0$ for the
% internal edge, and ${\bf b}_i$ for the edge pendant to taxon $i$.
% The rest, $x_1$ through $x_4$ are the four distances to the new taxon.
%
The rows label the extremal rays of the cone, and are divided into sections.
The first section, labeled with $c$, are the rays which give branch length/extra taxon distances with a nontrivial internal branch length for the lower tree. 
This is visible because of the $1$ in the ${\bf b}_0$ column.
These rays are interesting as they represent the ``minimal'' rogue taxon examples.
We analyze these $c_i$ in more detail below.

The second section, labeled with $e$, $f$, and $\shift$, shows how the pendant  (leading to a leaf) branch lengths of the lower tree and the distances to the new taxon can be modified without changing the upper tree.
That is, any positive multiple of these vectors can be added to a point in the
cone while staying in the same polyhedral cone.
For instance, $e_4$ says that we can increase the branch lengths ${\bf b}_1$ and ${\bf b}_3$ simultaneously while maintaining the same upper tree.
The ray $f_3$, for example, (which is all zero except for the $x_2$ column), says that we can increase the distance of the new taxon to the second original taxon without changing the upper tree.
The $\shift_i$ are simply the \emph{shift vectors} corresponding to the pendant branches. 
Thus $\shift_i$ means that we can increase the $i$th pendant branch length while increasing the distance of the new taxon to the $i$th original taxon without changing the upper tree.

These extremal rays can give some sufficient conditions for rogue taxon behavior.
We specify branch lengths of quartets by a vector giving branch lengths in the order $({\bf b}_0,\dots,{\bf b}_4)$.
We say that a vector $\mathbf{x}$ is a rogue vector for a branch length vector $\mathbf{b}$ if the BME tree for the combined data as in Table~\ref{tab:coneverts4} is the tree $(((1,5),3),(2,4))$.
We will call the cone given by positive linear combinations of the set
\[
\{ (0, 1, 1, 1, 1), (0, 1, 1, 1, 0), (0, 1, 0, 1, 1), (0, 1, 0, 1, 0), (0, 1, 0, 0, 0) \}
\]
the \emph{extension cone}.
Any element from this cone can be added to a branch length set without changing the polyhedral cone; this can be seen by looking at the $e_i$ vectors above.

Note that any vector satisfying $0 \le x_1 \le x_3 \le \min(x_2, x_4)$ sits
in the cone generated by the $f_i$ restricted to their last four
coordinates. 
Therefore we conclude:
\begin{proposition}
  Any vector satisfying $0 \le x_1 \le x_3 \le \min(x_2, x_4)$ is a rogue vector for any tree with branch length vector given by either $(1, 4, 0, 3, 3)$ or $(1, 3, 0, 3, 0)$ plus any element of the extension cone.
  \label{prop:rogueFunnyBL}
\end{proposition}

The next proposition gives rogue criteria for a quartet tree with
arbitrary internal branch length.
The proof is simple: just look at $c_5$ in Table~\ref{tab:coneverts4},
which shows that $(0,3,0,3)$ is a rogue vector for the quartet with
trivial pendant branch lengths and internal branch length 1.

\begin{proposition}
  Any quartet tree has a rogue vector with an entry greater than or equal to three times the internal branch length of the lower tree.
  \label{prop:rogueThreeTimes}
\end{proposition}

Although the above propositions do give some conditions on when the
rogue taxon effect appears for four taxa, they do not specify how likely are we to end up in a rogue taxon situation.
They also give no information about trees on larger number of taxa.
In the next section, we gain some intuition about these questions via simulation.

%%%%%%%%%%%%%%%%%%%%%%%%%%%%%%%%%%%%%%%%%%%%%%%%%%%%%%%%%%%%%%%%%%%%%%%%%%%%%
%%%%%%%%%%%%%%%%%%%%%%%%%%%%%%%%%%%%%%%%%%%%%%%%%%%%%%%%%%%%%%%%%%%%%%%%%%%%%

\subsection{Simulations}
\label{sec:simulations}
Here we describe simulations performed to better understand the rogue taxon effect as it might appear in biological data.
These simulations show that, at least for small numbers of taxa, the rogue taxon effect is common when the extra distances are chosen without reference to the original tree.
They also suggest that the effect gets worse as the number of taxa increases.

We assume a random distribution for the branch lengths and distances to the new taxon.
Such simulations are not the only way to address these sorts of questions.
Volume computations of, e.g., spheres intersected with our polyhedral
cones are in principle possible, but they do not seem to admit a
closed form solution.
Thus our understanding of such volumes still depends strongly on Monte Carlo simulations \citep{BMEPolytope}.
Furthermore, such a volume may give less practical information than simulation using a reasonable model of branch lengths.

To better understand the frequency with which the rogue taxon phenomenon can occur, we simulate using the exponential distribution.
Although a simple arbitrary choice, the exponential distribution is realistic enough to be a branch length prior for Bayesian phylogenetic inference \citep{mrBayesManual}.
For a given lower tree, we generate branch lengths for that tree according to the mean one exponential distribution, then generate distances to the extra taxon via the exponential distribution with mean equal to the expected pairwise distance between tips of the tree.
Then, we find the upper tree (i.e.\ the BME tree for the original data set plus the rogue taxon) and check to see how many bipartitions of the upper tree (restricted to the lower taxa) are not contained in the lower tree. 
This number is the Robinson-Foulds distance between the upper and lower trees used in Section~\ref{sec:funny-trees}.

%We call $t_4 = ((1,2),(3,4))$, $t_5 = ((1,2),3,(4,5))$, $t_{6a} = (((1,2),5),((3,4),6))$, and $t_{6b} = ((1,2),(3,4),(5,6))$.
    %\RF & $t_4$     & $t_5$      & $t_{6a}$    &  $t_{6b}$  \\
\begin{table}
  \begin{tabular}{c|ccccc}%{l|lllll}
   \RF & $((1,2),(3,4))$ & $((1,2),3,(4,5))$ & $(((1,2),5),((3,4),6))$ & $((1,2),(3,4),(5,6))$ \\
    \hline
    0  &  0.705071  &  0.502925  &  0.380863   &  0.381869  \\
    1  &  0.294929  &  0.364874  &  0.367523   &  0.363955  \\
    2  &  -         &  0.132201  &  0.195223   &  0.209066  \\
    3  &  -         &  -         &  0.0563907  &  0.04511   \\
  \end{tabular}
  \caption{
  Simulation results for $10^7$ exponentially distributed branch lengths and distances to rogue taxa.
  The columns are labeled by the topology of the lower tree.
  The numbers in the table represent the fraction of time that the corresponding Robinson-Foulds distance between the upper and lower trees appeared via the rogue taxon effect.
  }
  \label{tab:sim}
\end{table}

The results of $10^7$ exponentially drawn branch lengths are shown in Table~\ref{tab:sim}; it shows that a taxon added with random data can substantially alter the structure of the phylogenetic tree. 
Indeed, almost 30\% of the lifted four taxon trees do not contain the original topology, growing to almost 50\% for five taxa, then almost 62\% for the six taxon topologies.

We emphasize that such simulations do not paint an accurate picture of the rogue taxon effect for real data.
Indeed, even the worst data does not have completely random distances: even ``random'' sequence data will not have random distances to the
rest of the tree.
Nevertheless, we believe that these results indicate that this area merits further investigation and that the effective volume of these ``rogue'' polyhedral cones is not small.

In the reduced BME setting it can happen that multiple bifurcating upper trees are associated with a cone of the reduced normal fan for a given lower tree.
That is, the trees all have the same BME length for given lower tree branch lengths and rogue taxon distances.
We have observed in the example presented here that when there are these multiple trees, the Robinson-Foulds distance between the lower tree and these multiple upper trees (restricted to the lower taxa) for a given cone are equal.
It would be interesting to know if this is true in the general case.

The equivalent fact for the quartet distance is not true.
In the case of the lower tree being $(((1,2),5),((3,4),6))$, there is a cone of the reduced normal fan associated with both $(1,((((2,3),(4,7)),6),5))$ and $(1,((((2,6),(4,7)),3),5))$.
Restricting to the lower taxa, these trees are $(1,((((2,3),4),6),5))$ and $(1,((((2,6),4),3),5))$, which have quartet distances 10 and 11, respectively, to the lower tree.

%%%%%%%%%%%%%%%%%%%%%%%%%%%%%%%%%%%%%%%%%%%%%%%%%%%%%%%%%%%%%%%%%%%%%%%%%%%%%
%%%%%%%%%%%%%%%%%%%%%%%%%%%%%%%%%%%%%%%%%%%%%%%%%%%%%%%%%%%%%%%%%%%%%%%%%%%%%

\section{Conclusions and future directions}
\label{sec:conclusions}

We have investigated the effect of adding an extra ``rogue'' taxon into a phylogenetic data set for BME phylogenetic inference.
We have shown that rogue taxa can have significant though not arbitrary effects on the tree.
For a small number of taxa, we can delineate the domain of the rogue taxon effect.
Simulations show that the rogue taxon effect is very significant when the data for the rogue taxon is chosen randomly without reference to the topology of the original tree.

The results presented here may have algorithmic consequences for phylogenetic inference.
It is common for inference programs to start with a tree on three taxa then build a tree by adding taxa sequentially.
%This process is important enough to even merit research on the optimal order of taxon addition \citep{kimEaOptimizingTaxonAddition03}.
Software packages using sequential taxon addition, such as PHYLIP \citep{felsensteinPhylip95} and fastDNAml \citep{olsenFastDNAml94} do optimize the tree after addition using rearrangements; the question of strict sequential addition performance is still important in order to determine the amount of post-addition optimization required.
Furthermore, ``evolutionary placement algorithms'' for large amounts of sequence data have been proposed whereby a ``query'' sequences are inserted into a fixed ``reference tree'' \citep{vonMeringEaQuantitative08, bergerStamatakisEPA09}.
The accuracy of such algorithms compared to traditional phylogenetics algorithms can be seen as an aspect of the rogue taxon problem.

An interesting next direction would be to consider situations where rogue taxa
do not have arbitrary data, but appear via misspecified evolutionary
models.
  This will hopefully give a clearer understanding of the
actual impact of rogue taxa.
  It would also be interesting to see if
some of the results presented here also extend to other inference
criteria, such as parsimony or maximum likelihood.
  Some results, such
as the simulation results presented above, will certainly be different
in this new setting but others may correspond well.
  Maximum
likelihood and parsimony are considerably more difficult to analyze,
but hopefully the results presented here can act as a guide.

\section*{Acknowledgments}
\label{sec:acknowledgements}
We thank Tracy Heath for directing us to the taxon sampling debate, Mike Steel for simplifying the proof of Theorem~\ref{thm:funny}, and Bernd Sturmfels, Lior Pachter and Rudy Yoshida for fruitful discussions.
We are grateful to the two anonymous reviewers for their careful reading of the manuscript.

%%%%%%%%%%%%%%%%%%%%%%%%%%%%%%%%%%%%%%%%%%%%%%%%%%%%%%%%%%%%%%%%%%%%%%%%%%%%%
%%%%%%%%%%%%%%%%%%%%%%%%%%%%%%%%%%%%%%%%%%%%%%%%%%%%%%%%%%%%%%%%%%%%%%%%%%%%%

\bibliography{paperBME}
\bibliographystyle{plainnat}
\vspace{1cm}

\end{document}